\documentclass[12pt]{article}
\usepackage{amssymb}
\usepackage{amsfonts}
\usepackage{amsmath}
\usepackage{eurosym}
\usepackage{makeidx}
\usepackage{amssymb}
\usepackage{times}
\usepackage{anysize}
\usepackage{cite}
\usepackage{lineno}

\marginsize{2cm}{2cm}{1cm}{1cm}
\newtheorem{theorem}{Theorem}[section]

\newtheorem{condition}[theorem]{Condition}

\newtheorem{corollary}[theorem]{Corollary}

\newtheorem{lemma}[theorem]{Lemma}

\newenvironment{proof}[1][Proof]{\noindent\textbf{#1.} }{\ \rule{0.5em}{0.5em}}
\input{tcilatex}
\begin{document}

\title{Decay of Complex-time Determinantal and Pfaffian\ Correlation
Functionals in Lattices}
\author{N. J. B. Aza \and J.-B. Bru \and W. de Siqueira Pedra}
\maketitle

\begin{abstract}
We supplement the determinantal and Pfaffian bounds of \cite{Sims-Warzel}
for many-body localization of quasi-free fermions, by considering the high
dimensional case and complex-time correlations. Our proof uses the
analyticity of correlation functions via the Hadamard three-line theorem. We
show that the dynamical localization for the one-particle system yields the
dynamical localization for the many-point fermionic correlation functions,
with respect to the Hausdorff distance in the determinantal case. In \cite%
{Sims-Warzel}, a stronger notion of decay for many-particle configurations
was used but only at dimension one and for real times. Considering
determinantal and pfaffian correlation functionals for complex times is
important in the study of weakly interacting fermions.\bigskip

\bigskip

\end{abstract}

\section{Introduction}

Since a few years, the problem of (Anderson) localization in many-body
systems is garnering attention. The mathematical understanding of this
phenomenon for \emph{interacting }quantum particles, as adressed in 2006 by
\cite{interaction0} for weakly interacting fermions at small densities, is a
long-term goal. In 2009, \cite{interaction2,Aiz-Warzel} contributed first
rigorous results. In 2016, an exponential decay of many-particle
correlations was proven for quasi-free fermions in \emph{one-dimensional}
lattices with disorder \cite{Sims-Warzel}.\ Via the Jordan-Wigner
transformation, this includes the disordered $XY$ spin chains. This last
paper has attracted much attention and it has already been cited many times
in one and a half year. See, e.g., \cite%
{cite1,cite2,cite3,cite4,cite5,cite6,cite7,cite8,cite9}.

As pointed out\ in \cite{Sims-Warzel}, it is an interesting open question
(a) whether \cite[Theorems 1.1 and 1.2]{Sims-Warzel} can be generalized to
higher dimensions. Another open question (b) is their generalization for
\emph{complex}-time correlation functions. This\ last point is relevant
because such correlation functions (of quasi-free fermions) can be useful to
study localization of \emph{weakly interacting} fermion systems on lattices.
In fact, (quasi-free) complex-time correlation functions appear in the
perturbative expansion of (full) correlations for weakly interacting
systems. See, for instance, \cite[Section 5.4.1]{BratteliRobinson}.

By considering the many-body localization in the sense of the Hausdorff
distance, like in \cite{Aiz-Warzel}, we propose an answer to both questions
(a) and (b), using\ the Hadamard three-line theorem (Section \ref{Section
Hadamard}). See Corollary \ref{conditon local copy(2)}, which, together with
Theorem \ref{conditon local copy(1)}, is our\ main result on determinantal
correlation functionals. In a similar way, we also prove the decay of
complex-time Pfaffian\ correlation functionals with respect to a splitting
width (like \cite[Equation (5.9)]{Aiz-Warzel}) of particle configurations.
This is a version of \cite[Theorem 1.4]{Sims-Warzel} which holds true at any
dimension $d\in \mathbb{N}$. See Theorem \ref{thm majorama}.

\section{Decay of Determinantal Correlation Functionals}

\subsection{Setup of the Problem and Main Results\label{Section main}}

\noindent \underline{(i):} Let $d\in \mathbb{N}$. For a fixed parameter $%
\epsilon \in (0,1]$, we define%
\begin{equation}
\mathfrak{d}_{\epsilon }(\mathcal{X}_{1},\mathcal{X}_{2})\doteq \max \left\{
\max_{x_{1}\in \mathcal{X}_{1}}\min_{x_{2}\in \mathcal{X}_{2}}\left\vert
x_{1}-x_{2}\right\vert ^{\epsilon },\max_{x_{2}\in \mathcal{X}%
_{2}}\min_{x_{1}\in \mathcal{X}_{1}}\left\vert x_{1}-x_{2}\right\vert
^{\epsilon }\right\} ,\qquad \mathcal{X}_{1},\mathcal{X}_{2}\subset \mathbb{Z%
}^{d},  \label{Hausdorf}
\end{equation}%
which is the well-known Hausdorff distance between the two sets, associated
with the metric $(x_{1},x_{2})\mapsto |x_{1}-x_{2}|^{\epsilon }$ on $\mathbb{%
Z}^{d}$.\medskip

\noindent \underline{(ii):} We consider (non-relativistic) fermions in the
lattice $\mathbb{Z}^{d}$ with arbitrary finite spin set $\mathrm{S}$. Thus,
we define the one-particle Hilbert space to be $\mathfrak{h}\doteq \ell
^{2}\left( \mathbb{Z}^{d};\mathbb{C}^{\mathrm{S}}\right) $, the canonical
orthonormal basis $\left\{ \mathfrak{e}_{x,\sigma }\right\} _{(x,\sigma )\in
\mathbb{Z}^{d}\times \mathrm{S}}$ of which is%
\begin{equation}
\mathfrak{e}_{x_{0},\sigma _{0}}(x,\sigma )\doteq \delta _{x,x_{0}}\delta
_{\sigma ,\sigma _{0}}\ ,\qquad x,x_{0}\in \mathbb{Z}^{d},\ \sigma ,\sigma
_{0}\in \mathrm{S}.  \label{canonical onb1}
\end{equation}%
\mbox{ }\medskip

\noindent \underline{(iii):} Let $(\Omega ,\mathfrak{F},\mathfrak{a})$ be a
standard\footnote{%
I.e., $\mathfrak{F}$ is the Borel $\sigma $-algebra of a Polish space $%
\Omega $.} probability space. As is usual, $\mathbb{E}[\ \cdot \ ]$ denotes
the expectation value associated with the probability measure $\mathfrak{a}$%
. We consider $\mathfrak{F}$-measurable\ families $\{H_{\omega }\}_{\omega
\in \Omega }\subset \mathcal{B}\left( \mathfrak{h}\right) $ of bounded
one-particle Hamiltonians satisfying the following (one-body localization)
assumption, at fixed $\beta \in \mathbb{R}^{+}$:

\begin{condition}
\label{conditon local}\mbox{ }\newline
There is a Borel set $I\subset \mathbb{R}$ as well as constants $\epsilon
\in (0,1]$, $D$ and $\mu \in \mathbb{R}^{+}$ such that, for all $x_{1}\in
\mathbb{Z}^{d}$ and $R>0$,
\begin{equation}
\sum\limits_{x_{2}\in \mathbb{Z}^{d}:\left\vert x_{1}-x_{2}\right\vert
^{\epsilon }\geq R}\mathbb{E}\left[ \sup_{z\in \mathbb{S}_{\beta
}}\max_{\sigma _{1},\sigma _{2}\in \mathrm{S}}\left\vert \left\langle
\mathfrak{e}_{x_{1},\sigma _{1}},\frac{\mathrm{e}^{izH_{\omega }}\chi
_{I}\left( H_{\omega }\right) }{1+\mathrm{e}^{\beta H_{\omega }}}\mathfrak{e}%
_{x_{2},\sigma _{2}}\right\rangle _{\mathfrak{h}}\right\vert \right] \leq D\,%
\mathrm{e}^{-\mu R},  \label{truc}
\end{equation}%
where
\begin{equation}
\mathbb{S}_{\beta }\doteq \mathbb{R}-i\left[ 0,\beta \right] ,\qquad \beta
\in \mathbb{R}^{+},  \label{strip}
\end{equation}%
$\chi _{I}$ is the characteristic function of the set $I$, and $\left\vert
x_{1}-x_{2}\right\vert $ the euclidean distance between the lattice points $%
x_{1},x_{2}\in \mathbb{Z}^{d}$.
\end{condition}

This assumption is similar to the so-called strong exponential dynamical
localization in $I$, see, e.g., \cite[Definition 7.1]{WAlivre}. Note that,
for $\epsilon \in (0,1]$, $(x_{1},x_{2})\mapsto \left\vert
x_{1}-x_{2}\right\vert ^{\epsilon }$ defines a translation invariant metric
on the lattice $\mathbb{Z}^{d}$. Observe also that, for all $\beta \in
\mathbb{R}^{+}$ and $z\in \mathbb{S}_{\beta }$, the function $\lambda
\mapsto |\mathrm{e}^{z\lambda }\left( 1+\mathrm{e}^{\beta \lambda }\right)
^{-1}|$ on $\mathbb{R}$ is bounded by $1$. In particular, the left-hand side
of (\ref{truc}) is bounded by the eigenfunction correlator \cite[Eq. (7.1)]%
{WAlivre}. Condition \ref{conditon local} replaces \cite[Eq. (1.19)]%
{Sims-Warzel}, noting that
\begin{equation}
\rho \left( s,t\right) =\frac{\mathrm{e}^{i\left( t-s\right) H_{\omega }}}{1+%
\mathrm{e}^{\beta H_{\omega }}},\qquad s,t\in \mathbb{R},  \label{rho}
\end{equation}%
is the main example they have in mind \cite[Eq. (2.37)]{Sims-Warzel}.
\medskip

\noindent \underline{(iv):} Let $\mathrm{CAR}(\mathfrak{h})$ be the CAR $%
C^{\ast }$--algebra generated by the unit $\mathbf{1}$ and $\{a(\varphi
)\}_{\varphi \in \mathfrak{h}}$. For any $A_{1},A_{2}\in \mathrm{CAR}(%
\mathfrak{h})$ and any $z_{1},z_{2}\in \mathbb{C}$, we define%
\begin{equation}
\mathbb{O}_{z_{1},z_{2}}\left( A_{1},A_{2}\right) \doteq \left\{
\begin{array}{ccc}
A_{1}A_{2} & \text{if} & \mathrm{Im}\left( z_{1}\right) \leq \mathrm{Im}%
\left( z_{2}\right) , \\
-A_{2}A_{1} & \text{if} & \mathrm{Im}\left( z_{1}\right) >\mathrm{Im}\left(
z_{2}\right) .%
\end{array}%
\right.  \label{O alacon}
\end{equation}%
\mbox{ }\medskip

\noindent \underline{(v):} For any $\beta \in \mathbb{R}^{+}$ and $\omega
\in \Omega $, we define the (gauge invariant) quasi-free state $\rho
_{\omega }\equiv \rho _{\beta ,\omega }$ by the condition%
\begin{equation}
\rho _{\omega }\left( a(\varphi _{1})^{\ast }a(\varphi _{2})\right)
=\left\langle \varphi _{2},\frac{1}{1+\mathrm{e}^{\beta H_{\omega }}}\varphi
_{1}\right\rangle _{\mathfrak{h}}\ ,\qquad \varphi _{1},\varphi _{2}\in
\mathfrak{h}.  \label{conditions}
\end{equation}%
This state is the unique KMS state at inverse temperature $\beta \in \mathbb{%
R}^{+}$ associated with the unique strongly continuous group $\{\tau
_{t}^{(\omega )}\}_{t\in {\mathbb{R}}}$ of (Bogoliubov) automorphisms of $%
\mathrm{CAR}(\mathfrak{h})$ satisfying%
\begin{equation}
\tau _{t}^{(\omega )}\left( a\left( \varphi \right) \right) =a(\mathrm{e}%
^{itH_{\omega }}\varphi )\ ,\qquad t\in {\mathbb{R}},\ \varphi \in \mathfrak{%
h}.  \label{bogo transformation}
\end{equation}%
Note that, for all $\varphi \in \mathfrak{h}$, the maps
\begin{equation*}
t\mapsto \tau _{t}^{(\omega )}\left( a\left( \varphi \right) \right) \quad
\text{and}\quad t\mapsto \tau _{t}^{(\omega )}\left( a\left( \varphi \right)
^{\ast }\right)
\end{equation*}%
on $\mathbb{R}$ uniquely extend to entire functions: For any $z\in \mathbb{C}
$ and $\varphi \in \mathfrak{h}$,
\begin{equation}
\tau _{z}^{(\omega )}\left( a\left( \varphi \right) ^{\ast }\right) \doteq a(%
\mathrm{e}^{izH_{\omega }}\varphi )^{\ast }\quad \text{and}\quad \tau
_{z}^{(\omega )}\left( a\left( \varphi \right) \right) \doteq a(\mathrm{e}^{i%
\overline{z}H_{\omega }}\varphi ).  \label{bogo transformation2}
\end{equation}%
Observe additionally that, for any $z_{1},z_{2}\in \mathbb{C}$ and $\varphi
_{1},\varphi _{2}\in \mathfrak{h}$,
\begin{eqnarray}
&&\rho _{\omega }\left( \mathbb{O}_{z_{1},z_{2}}\left( \tau
_{z_{1}}^{(\omega )}(a(\varphi _{1})^{\ast }),\tau _{z_{2}}^{(\omega
)}(a(\varphi _{2}))\right) \right)  \label{eq sup} \\
&=&\left\{
\begin{array}{ccc}
\left\langle \varphi _{2},\frac{\mathrm{e}^{i\left( z_{1}-z_{2}\right)
H_{\omega }}}{1+\mathrm{e}^{\beta H_{\omega }}}\varphi _{1}\right\rangle _{%
\mathfrak{h}} & \text{if} & \mathrm{Im}\left( z_{1}\right) \leq \mathrm{Im}%
\left( z_{2}\right) , \\
-\left\langle \varphi _{2},\frac{\mathrm{e}^{\left( \beta
+i(z_{1}-z_{2})\right) H_{\omega }}}{1+\mathrm{e}^{\beta H_{\omega }}}%
\varphi _{1}\right\rangle _{\mathfrak{h}} & \text{if} & \mathrm{Im}\left(
z_{1}\right) >\mathrm{Im}\left( z_{2}\right) .%
\end{array}%
\right.  \notag
\end{eqnarray}%
\mbox{ }\medskip

Below, we show that strong one-body localization, in the sense of Condition %
\ref{conditon local}, yields the corresponding many-body localization for
the quasi-free state $\rho _{\omega }$, in the sense of the Hausdorff
distance, as stated in Corollary \ref{conditon local copy(2)}. This is
achieved by estimating, in Theorem \ref{determinant bound copy(1)},
determinants of the form
\begin{equation}
\det \left[ G_{\omega }\left( \left( \varphi _{k},z_{k}\right) ,\left(
\varphi _{N+l},z_{N+l}\right) \right) \right] _{k,l=1}^{N}  \label{**}
\end{equation}%
in terms of the entries of one single row or column. In (\ref{**}), $\beta
\in \mathbb{R}^{+}$, $N\in \mathbb{N}$, $\varphi _{1},\ldots ,\varphi
_{2N}\in \mathfrak{h}$ are normalized vectors, $z_{1},\ldots ,z_{2N}\in
\mathbb{S}_{\beta }$ and%
\begin{equation*}
G_{\omega }\left( \left( \varphi _{k},z_{k}\right) ,\left( \varphi
_{N+l},z_{N+l}\right) \right) \doteq \rho _{\omega }\left( \mathbb{O}%
_{z_{k},z_{N+l}}\left( \tau _{z_{k}}^{(\omega )}(a(\varphi _{k})^{\ast
}),\tau _{z_{N+l}}^{(\omega )}(a(\varphi _{N+l}))\right) \right)
\end{equation*}%
is the two-point, complex-time-ordered correlation function associated with
the quasi-free state $\rho _{\omega }$.

\begin{theorem}
\label{conditon local copy(1)}\mbox{ }\newline
Let $\{H_{\omega }\}_{\omega \in \Omega }\subset \mathcal{B}\left( \mathfrak{%
h}\right) $ be a family of bounded Hamiltonians. For all $\omega \in \Omega $%
, $\beta \in \mathbb{R}^{+}$, $N\in \mathbb{N}$, norm-one vectors $\varphi
_{1},\ldots ,\varphi _{2N}\in \mathfrak{h}$, and $z_{1},\ldots ,z_{2N}\in
\mathbb{S}_{\beta }$ (see (\ref{strip}))%
\begin{eqnarray*}
&&\left\vert \det \left[ G_{\omega }\left( \left( \varphi _{k},z_{k}\right)
,\left( \varphi _{N+l},z_{N+l}\right) \right) \right] _{k,l=1}^{N}\right\vert
\\
&\leq &\min \left\{ \min_{k\in \left\{ 1,\ldots ,N\right\}
}\sum_{l=1}^{N}\left\vert G_{\omega }\left( \left( \varphi _{k},z_{k}\right)
,\left( \varphi _{N+l},z_{N+l}\right) \right) \right\vert ,\min_{l\in
\left\{ 1,\ldots ,N\right\} }\sum_{k=1}^{N}\left\vert G_{\omega }\left(
\left( \varphi _{k},z_{k}\right) ,\left( \varphi _{N+l},z_{N+l}\right)
\right) \right\vert \right\} .
\end{eqnarray*}
\end{theorem}

\begin{proof}
Fix all parameters of the theorem. By expanding the determinant along a
fixed row or column, for any $m\in \left\{ 1,\ldots ,N\right\} $,
\begin{align}
& \det \left[ G_{\omega }\left( \left( \varphi _{k},z_{k}\right) ,\left(
\varphi _{N+l},z_{N+l}\right) \right) \right] _{k,l=1}^{N}  \label{laplace0}
\\
& =\sum_{n=1}^{N}\left( -1\right) ^{m+n}G_{\omega }\left( \left( \varphi
_{m},z_{m}\right) ,\left( \varphi _{N+n},z_{N+n}\right) \right)   \notag \\
& \times \det \left[ G_{\omega }\left( \left( \varphi _{k},z_{k}\right)
,\left( \varphi _{N+l},z_{N+l}\right) \right) \right] _{\substack{ k\in
\left\{ 1,\ldots ,N\right\} \backslash \left\{ m\right\}  \\ l\in \left\{
1,\ldots ,N\right\} \backslash \left\{ n\right\} }}  \notag \\
& =\sum_{n=1}^{N}\left( -1\right) ^{m+n}G_{\omega }\left( \left( \varphi
_{n},z_{n}\right) ,\left( \varphi _{N+m},z_{N+m}\right) \right)   \notag \\
& \times \det \left[ G_{\omega }\left( \left( \varphi _{k},z_{k}\right)
,\left( \varphi _{N+l},z_{N+l}\right) \right) \right] _{\substack{ k\in
\left\{ 1,\ldots ,N\right\} \backslash \left\{ n\right\}  \\ l\in \left\{
1,\ldots ,N\right\} \backslash \left\{ m\right\} }}.  \notag
\end{align}%
Then, the assertion directly follows from Lemma \ref{determinant bound}.
\end{proof}

\begin{corollary}
\label{conditon local copy(2)}\mbox{ }\newline
If Condition \ref{conditon local} holds true then, for all $\beta \in
\mathbb{R}^{+}$, $N\in \mathbb{N}$, $\mathcal{X}_{1}=\{x_{1},\ldots ,x_{N}\},%
\mathcal{X}_{2}=\{x_{N+1},\ldots ,x_{2N}\}\subset \mathbb{Z}^{d}$ such that $%
\left\vert \mathcal{X}_{1}\right\vert =\left\vert \mathcal{X}_{2}\right\vert
=N$, and $z_{1},\ldots ,z_{2N}\in \mathbb{S}_{\beta }$,
\begin{equation*}
\mathbb{E}\left[ \max_{\sigma _{1},\ldots ,\sigma _{2N}}\left\vert \mathrm{%
det}\left[ G_{\omega }\left( (\chi _{I}(H_{\omega })\mathfrak{e}%
_{x_{k},\sigma _{k}},z_{k}),(\chi _{I}(H_{\omega })\mathfrak{e}%
_{x_{N+l},\sigma _{N+l}},z_{N+l})\right) \right] _{k,l=1}^{N}\right\vert %
\right] \leq D\,\mathrm{e}^{-\mu \mathfrak{d}_{\epsilon }(\mathcal{X}_{1},%
\mathcal{X}_{2})},
\end{equation*}%
where $\mathfrak{d}_{\epsilon }(\mathcal{X}_{1},\mathcal{X}_{2})$ is the
Hausdorff distance (\ref{Hausdorf}) between the $N$-particle configurations $%
\mathcal{X}_{1}$ and $\mathcal{X}_{2}$. Recall that $\chi _{I}$ is the
characteristic function of the Borel set $I$ and note that the constants $%
\epsilon $, $D$ and $\mu $ are exactly the same as in Condition \ref%
{conditon local}.
\end{corollary}

\begin{proof}
Combine Condition \ref{conditon local} and Theorem \ref{conditon local
copy(1)} with Equations (\ref{bogo transformation2}) and (\ref{eq sup}).
\end{proof}

The analogue of \cite[Theorem 1.1]{Sims-Warzel}, i.e., an estimate like
Corollary \ref{conditon local copy(2)} for the many-point correlation
functions at fixed $\omega \in \Omega $, instead of an estimate for their
expectation values, easily follows by replacing Condition \ref{conditon
local} with a similar bound for a fixed $\omega \in \Omega $. We omit the
details.

The bound of Corollary \ref{conditon local copy(2)}\ is a version of \cite[%
Theorem 1.2]{Sims-Warzel} which holds at any dimension $d\in \mathbb{N}$ and
for any complex times within the strip $\mathbb{S}_{\beta }$. However, two
observations in relation with \cite{Sims-Warzel} are important to mention:

\begin{itemize}
\item Since, for any $\mathcal{X}_{1},\mathcal{X}_{2},\mathcal{Y}_{1},%
\mathcal{Y}_{2}\subset \mathbb{Z}^{d}$,
\begin{equation*}
\mathfrak{d}_{\epsilon }(\mathcal{X}_{1}\cup \mathcal{X}_{2},\mathcal{Y}%
_{1}\cup \mathcal{Y}_{2})\leq \max \left\{ \mathfrak{d}_{\epsilon }(\mathcal{%
X}_{1},\mathcal{Y}_{1}),\mathfrak{d}_{\epsilon }(\mathcal{X}_{2},\mathcal{Y}%
_{2})\right\} ,
\end{equation*}%
we have%
\begin{equation}
\mathfrak{d}_{\epsilon }(\mathcal{X},\mathcal{Y})\leq \mathfrak{d}_{\epsilon
}^{(S)}(\mathcal{X},\mathcal{Y})\doteq \min_{\pi \in \mathcal{S}%
_{N}}\max_{j\in \{1,\ldots ,N\}}\left\vert x_{j}-y_{\pi \left( j\right)
}\right\vert ^{\epsilon }  \label{hadamabis}
\end{equation}%
for any set $\mathcal{X}=\{x_{1},\ldots ,x_{N}\}\subset \mathbb{Z}^{d}$ and $%
\mathcal{Y}=\{y_{1},\ldots ,y_{N}\}\subset \mathbb{Z}^{d}$ of $N\in \mathbb{N%
}$ (different) lattice points. Here, $\mathcal{S}_{N}$ is the set of all
permutations $\pi $ of $N$ elements. The distance we use, i.e., the
Hausdorff distance (\ref{Hausdorf}), is therefore weaker than the
symmetrized configuration distance $\mathfrak{d}_{\epsilon }^{(S)}$ \cite[%
Equation (1.13) and remarks below]{Sims-Warzel}. Nevertheless, Corollary \ref%
{conditon local copy(2)} yields the main features of localization. Whether
Corollary \ref{conditon local copy(2)} holds true, at any dimension, when $%
\mathfrak{d}_{\epsilon }$ is replaced with $\mathfrak{d}_{\epsilon }^{(S)}$
is an open question. See also discussions of \cite[Section 1.3]{Aiz-Warzel}.

\item The proofs of \cite[Theorems 1.1 and 1.2]{Sims-Warzel} use that, for
all $N\in \mathbb{N}$, $x_{1},\ldots ,x_{2N}\in \mathbb{Z}^{d}$, $\sigma
_{1},\ldots ,\sigma _{2N}\in \mathrm{S}$, and $t_{1},\ldots ,t_{2N}\in
\mathbb{R}$, the $N\times N$ matrix
\begin{equation*}
\mathbf{M}\doteq \left[ \left\langle \mathfrak{e}_{x_{N+l},\sigma
_{N+l}},\rho \left( t_{N+l},t_{k}\right) \mathfrak{e}_{x_{k},\sigma
_{k}}\right\rangle _{\mathfrak{h}}\right] _{k,l=1}^{N}
\end{equation*}%
defines an operator on $\mathbb{C}^{N}$ of norm at most $1$. This is true
even for complex times, provided that
\begin{equation}
z_{1}=\cdots =z_{N}\in \mathbb{S}_{\beta },\quad z_{N+1}=\cdots =z_{2N}\in
\mathbb{S}_{\beta },\quad \mathrm{Im}\left( z_{N}\right) \leq \mathrm{Im}%
\left( z_{N+1}\right) .  \label{condition time}
\end{equation}%
(cf. \cite[Erratum]{Sims-Warzel}).\ However, this is generally not true when
$z_{1},\ldots ,z_{2N}\in \mathbb{S}_{\beta }$ are different from each other.
For this reason, instead of a bound on the norm of $\mathbf{M}$, our proof
uses (in an essential way) the analyticity of correlation functions with
respect to complex times.
\end{itemize}

The results of this section are also reminiscent of \cite[Theorem 1.1]%
{Aiz-Warzel} where a bound like Corollary \ref{conditon local copy(2)}, with
the Hausdorff distance but for complex times satisfying (\ref{condition time}%
), can be found for $n$-particle\ correlation functions. Note, additionally,
that in \cite{Aiz-Warzel} a particle interaction is included, but \emph{no}
particle statistics is taken into account: The $n$-particle Hilbert space is
the full space $\ell ^{2}\left( \mathbb{Z}^{nd}\right) $. By contrast, we
consider many-\emph{fermion} systems, which would correspond in \cite[%
Theorem 1.1]{Aiz-Warzel} to restrict $\ell ^{2}\left( \mathbb{Z}^{nd}\right)
$ to its subspace of antisymmetric functions. In this situation, the
one-particle localization theory cannot be directly used, even in the
quasi-free fermion case. Moreover, we do not fix the particle number, by
using the grand-canonical setting.

Finally, observe that \emph{quasi-free}, complex-time-ordered, many-point
correlations appear in the perturbative expansion of\ \emph{interacting}
correlation functions. See, e.g., \cite[Section 5.4.1]{BratteliRobinson}.
Therefore, as a first step towards the proof of localization in fully
interacting fermion systems, it is important to establish localization for
these correlations, as stated in Corollary \ref{conditon local copy(2)}. For
instance, by combining Corollary \ref{conditon local copy(2)} with \cite[%
Theorem 5.4.4]{BratteliRobinson}, one can show that a local, weak
interaction cannot destroy the (static) localization of the thermal,
many-point correlation functions of quasi-free fermions in lattices.

\subsection{Universal Bounds on Determinants from the Hadamard Three-line
Theorem}

For any permutation $\pi \in \mathcal{S}_{n}$ of $n\in {\mathbb{N}}$
elements with sign $(-1)^{\pi }$, we define the monomial $\mathbb{O}_{\pi
}(A_{1},\ldots ,A_{n})\in \mathrm{CAR}(\mathfrak{h})$ in $A_{1},\ldots
,A_{n}\in \mathrm{CAR}(\mathfrak{h})$ by the product%
\begin{equation}
\mathbb{O}_{\pi }\left( A_{1},\ldots ,A_{n}\right) \doteq \left( -1\right)
^{\pi }A_{\pi ^{-1}(1)}\cdots A_{\pi ^{-1}(n)}.  \label{Opi1}
\end{equation}%
In other words, $\mathbb{O}_{\pi }$ places the operator $A_{k}$ at the $\pi
(k)$th position in the monomial $(-1)^{\pi }A_{\pi ^{-1}(1)}\cdots A_{\pi
^{-1}(n)}$. Further, for all $k,l\in \{1,\ldots ,n\}$, $k\neq l$,%
\begin{equation}
\pi _{k,l}:\{1,2\}\rightarrow \{1,2\}  \label{Opi2}
\end{equation}%
is the identity function if $\pi (k)<\pi (l)$, otherwise $\pi _{k,l}$
interchanges $1$ and $2$. (Remark that $\mathbb{O}_{z_{k},z_{l}}$ (\ref{O
alacon}) is equal to $\mathbb{O}_{\pi _{k,l}}$ for a conveniently chosen
permutation $\pi $.) Then, the following identities holds true for
quasi--free states:

\begin{lemma}
\label{lemma exp copy(2)}\mbox{ }\newline
Let $\rho $ be a quasi--free state on $\mathrm{CAR}(\mathfrak{h})$. For any $%
N\in \mathbb{N}$, all permutations $\pi \in \mathcal{S}_{2N}$ and $\varphi
_{1},\ldots ,\varphi _{2N}\in \mathfrak{h}$,%
\begin{eqnarray}
&&\det \left[ \rho
\Big(%
\mathbb{O}_{\pi _{k,N+l}}\left( a(\varphi _{k})^{\ast },a(\varphi
_{N+l})\right)
\Big)%
\right] _{k,l=1}^{N}  \notag \\
&=&\rho
\Big(%
\mathbb{O}_{\pi }\left( a(\varphi _{1})^{\ast },\ldots ,a(\varphi
_{N})^{\ast },a(\varphi _{2N}),\ldots ,a(\varphi _{N+1})\right)
\Big)%
.  \label{ass det O}
\end{eqnarray}
\end{lemma}

\begin{proof}
See \cite[Lemma 3.1]{Bru-Pedra-deter-bound}. Compare with (\ref{ass O0-00bis}%
).
\end{proof}

Using Lemma \ref{lemma exp copy(2)} and the Hadamard three-line theorem (via
Corollary \ref{determinant bound copy(2)}), we obtain a universal bound on
determinants of the form (\ref{**}):

\begin{lemma}
\label{determinant bound}\mbox{ }\newline
Fix $H=H^{\ast }\in \mathcal{B}\left( \mathfrak{h}\right) $. Let the
quasi--free state $\rho $ on $\mathrm{CAR}(\mathfrak{h})$ be the unique KMS
state at inverse temperature $\beta \in \mathbb{R}^{+}$ associated with the
unique strongly continuous group $\{\tau _{t}\}_{t\in {\mathbb{R}}}$ of auto%
\-%
morphisms of $\mathrm{CAR}(\mathfrak{h})$ satisfying (\ref{bogo
transformation})-(\ref{bogo transformation2}) for $H_{\omega }=H$. Then, for
any $N\in \mathbb{N}$, $\varphi _{1},\ldots ,\varphi _{2N}\in \mathfrak{h}$
and $z_{1},\ldots ,z_{2N}\in \mathbb{S}_{\beta }$ (see (\ref{strip})),
\begin{equation*}
\left\vert \det \left[ \rho
\Big(%
\mathbb{O}_{z_{k},z_{N+l}}\left( \tau _{z_{k}}(a(\varphi _{k})^{\ast }),\tau
_{z_{N+l}}(a(\varphi _{N+l}))\right)
\Big)%
\right] _{k,l=1}^{N}\right\vert \leq \prod\limits_{k=1}^{2N}\left\Vert
\varphi _{k}\right\Vert _{\mathfrak{h}}.
\end{equation*}
\end{lemma}

\begin{proof}
Fix all parameters of the lemma and choose any permutation $\pi \in \mathcal{%
S}_{2N}$ such that, for all $k,l\in \left\{ 1,\ldots ,N\right\} $,
\begin{equation}
\mathrm{Im}(z_{k})\leq \mathrm{Im}(z_{N+l})\Leftrightarrow \pi \left(
k\right) <\pi \left( N+l\right) \ .  \label{p1}
\end{equation}%
Then, by Lemma \ref{lemma exp copy(2)},%
\begin{eqnarray}
&&\det \left[ \rho
\Big(%
\mathbb{O}_{z_{k},z_{N+l}}\left( \tau _{z_{k}}(a(\varphi _{k})^{\ast }),\tau
_{z_{N+l}}(a(\varphi _{N+l}))\right)
\Big)%
\right] _{k,l=1}^{N}  \label{by toto} \\
&=&\rho
\Big(%
\mathbb{O}_{\pi }\left( \tau _{z_{1}}(a(\varphi _{1})^{\ast }),\ldots ,\tau
_{z_{N}}(a(\varphi _{N})^{\ast }),\tau _{z_{2N}}(a(\varphi _{2N})),\ldots
,\tau _{z_{N+1}}(a(\varphi _{N+1}))\right)
\Big)%
.  \notag
\end{eqnarray}%
Define the entire analytic map $\Upsilon $ from $\mathbb{C}^{2N}$ to $%
\mathbb{C}$ by%
\begin{align}
\Upsilon (\xi _{1},\ldots ,\xi _{2N})& \doteq \rho
\Big(%
\mathbb{O}_{\pi }%
\Big(%
\tau _{\xi _{1}+\cdots +\xi _{2N-\pi \left( 1\right) +1}}(a(\varphi
_{1})^{\ast }),\ldots ,\tau _{\xi _{1}+\cdots +\xi _{2N-\pi \left( N\right)
+1}}(a(\varphi _{N})^{\ast }),  \notag \\
& \tau _{\xi _{1}+\cdots +\xi _{2N-\pi \left( 2N\right) +1}}(a(\varphi
_{2N})),\ldots ,\tau _{\xi _{1}+\cdots +\xi _{2N-\pi \left( N+1\right)
+1}}(a(\varphi _{N+1}))%
\Big)%
\Big)%
.  \label{gamma}
\end{align}%
Now, impose additionally that the permutation $\pi $ of $2N$ elements used
in (\ref{by toto})-(\ref{gamma}) satisfies, for any $k,l\in \left\{ 1,\ldots
,N\right\} $, $k\neq l$, the conditions
\begin{equation*}
\mathrm{Im}(z_{k})<\mathrm{Im}(z_{l})\Leftrightarrow \pi \left( k\right)
<\pi \left( l\right) ;\ \mathrm{Im}(z_{2N-k})<\mathrm{Im}(z_{2N-l})%
\Leftrightarrow \pi \left( 2N-k\right) <\pi \left( 2N-l\right) .
\end{equation*}%
Ergo, by (\ref{p1}),%
\begin{equation}
\mathrm{Im}(z_{\pi ^{-1}(1)})\leq \cdots \leq \mathrm{Im}(z_{\pi
^{-1}(N)})\leq \mathrm{Im}(z_{\pi ^{-1}(2N)})\leq \cdots \leq \mathrm{Im}%
(z_{\pi ^{-1}(N+1)})  \label{permutation}
\end{equation}%
and, by (\ref{by toto})-(\ref{gamma}), the assertion follows if we can bound
the function $\Upsilon $ on the tube $\mathfrak{T}_{2N}$ defined below by (%
\ref{tubetueb}) for $n=2N$. Since $\Upsilon $ is uniformally bounded on $%
\mathfrak{T}_{2N}$, it suffices to bound the function $\Upsilon $ on the
boundary%
\begin{equation*}
\partial \mathfrak{T}_{2N}\doteq \left\{ \left( \xi _{1},\ldots ,\xi
_{2N}\right) \in \mathbb{C}^{2N}:\forall j\in \{1,\ldots ,2N\},\ \mathrm{Im}%
(\xi _{j})\in \left\{ -\beta ,0\right\} ,\text{ }\sum_{j=1}^{2N}\mathrm{Im}%
(\xi _{j})\in \left\{ -\beta ,0\right\} \right\} ,
\end{equation*}%
by Corollary \ref{determinant bound copy(2)}. By the KMS property \cite[%
Section 5.3.1]{BratteliRobinson}, note that, for all $t_{1},\ldots
,t_{2N}\in \mathbb{R}$ and $k\in \left\{ 1,\ldots ,2N\right\} $,
\begin{equation*}
\Upsilon (t_{1},\ldots ,t_{k-1},t_{k}-i\beta ,t_{k+1},\ldots
,t_{2N})=\Upsilon (t_{k+1},\ldots ,t_{2N},t_{1},\ldots ,t_{k})
\end{equation*}%
while
\begin{equation*}
\sup_{(\xi _{1},\ldots ,\xi _{2N})\in \mathbb{R}^{2N}}\left\vert \Upsilon
(\xi _{1},\ldots ,\xi _{2N})\right\vert \leq
\prod\limits_{k=1}^{2N}\left\Vert \varphi _{k}\right\Vert _{\mathfrak{h}}.
\end{equation*}%
As a consequence,
\begin{equation}
\sup_{(\xi _{1},\ldots ,\xi _{2N})\in \mathfrak{T}_{2N}}\left\vert \Upsilon
(\xi _{1},\ldots ,\xi _{2N})\right\vert =\sup_{(\xi _{1},\ldots ,\xi
_{2N})\in \partial \mathfrak{T}_{2N}}\left\vert \Upsilon (\xi _{1},\ldots
,\xi _{2N})\right\vert \leq \prod\limits_{k=1}^{2N}\left\Vert \varphi
_{k}\right\Vert _{\mathfrak{h}}  \label{Hoelder}
\end{equation}%
and the assertion follows from (\ref{by toto}), (\ref{gamma}) and (\ref%
{tubetueb}).
\end{proof}

Observe that estimates like (\ref{Hoelder}) are related to the
generalization of the H\"{o}lder inequality to non--commutative $L^{p}$%
--spaces. See, e.g., \cite{Araki-Masuda}.

\section{Decay of Pfaffian Correlation Functionals}

An estimate similar to Corollary \ref{conditon local copy(2)} can be
obtained for Pfaffians of the two-point correlation functions on the $d$%
-dimensional square lattice $\mathbb{Z}^{d}$, by the same methods, because
they\ also can be seen, like in the proof of Lemma \ref{determinant bound},
as many-point correlation functions of quasi-free fermions.\medskip

\noindent \underline{(i):} For a fixed parameter $\epsilon \in (0,1]$ and
any subset $\mathcal{X}\subset \mathbb{Z}^{d}$ we define the quantity%
\begin{equation}
\ell _{\epsilon }(\mathcal{X})\doteq \max_{x\in \mathcal{X}}\min_{y\in
\mathcal{X}\backslash \left\{ x\right\} }\left\vert x-y\right\vert
^{\epsilon }.  \label{distance concentration}
\end{equation}%
It is a kind of splitting width of the configuration $\mathcal{X}$ with
respect to the metric $(x,y)\mapsto |x-y|^{\epsilon }$: This quantity is
large whenever isolated points of $\mathcal{X}$ are spread in space, but it
stays small if the points are packed in clusters containing at least two
points. It is used here to quantify the localization of Pfaffian correlation
functionals. Observe that $\ell _{\epsilon }$ is similar to the splitting
width of a configuration defined by \cite[Equation (5.9)]{Aiz-Warzel}%
.\medskip

\noindent \underline{(ii):} For any $N\in \mathbb{N}$, the Pfaffian of a $%
2N\times 2N$ \emph{skew-symmetric} complex matrix $M$ is defined by%
\begin{equation}
\mathrm{Pf}\left[ M_{k,l}\right] _{k,l=1}^{2N}\doteq \frac{1}{2^{N}N!}%
\sum_{\pi \in \mathcal{S}_{2N}}\left( -1\right) ^{\pi
}\prod\limits_{j=1}^{N}M_{\pi \left( 2j-1\right) ,\pi \left( 2j\right) },
\label{pgsdgkljsdlkfj}
\end{equation}
where we recall that $\mathcal{S}_{2N}$ is the set of all permutations of $2N
$ elements.\medskip

\noindent \underline{(iii):} Let the field operators be defined by
\begin{equation*}
\mathrm{B}\left( \varphi \right) \doteq a\left( \varphi \right) ^{\ast
}+a\left( \varphi \right) ,\qquad \varphi \in \mathfrak{h}.
\end{equation*}%
For $(x,\sigma )\in \mathbb{Z}^{d}\times \mathrm{S}$ and $\varphi =\mathfrak{%
e}_{x,\sigma }$ or $\varphi =i\mathfrak{e}_{x,\sigma }$, we obtain the
on-site Majorana fermions of \cite[Equation (1.22)]{Sims-Warzel}. \mbox{ }%
\medskip

Below, we show that strong one-body localization, in the sense of Condition %
\ref{conditon local}, yields the localization of many-point correlations of
field operators with respect to the quantity (\ref{distance concentration}).
This is achieved by estimating, in Theorem \ref{thm majorama}, Pfaffians of
the form
\begin{equation}
\mathrm{Pf}\left[ \mathcal{G}_{\omega }\left( \left( \varphi
_{k},z_{k}\right) ,\left( \varphi _{l},z_{l}\right) \right) \right]
_{k,l=1}^{2N}  \label{***}
\end{equation}%
in terms of the entries of one single row. In (\ref{***}), $\beta \in
\mathbb{R}^{+}$, $N\in \mathbb{N}$, $\varphi _{1},\ldots ,\varphi _{2N}\in
\mathfrak{h}$ are normalized vectors, $z_{1},\ldots ,z_{2N}\in \mathbb{S}%
_{\beta }$ and%
\begin{equation*}
\mathcal{G}_{\omega }\left( \left( \varphi _{k},z_{k}\right) ,\left( \varphi
_{l},z_{l}\right) \right) \doteq \rho _{\omega }\left( \mathbb{O}%
_{z_{k},z_{l}}\left( \tau _{z_{k}}^{(\omega )}(\mathrm{B}(\varphi
_{k})),\tau _{z_{l}}^{(\omega )}(\mathrm{B}(\varphi _{l}))\right) \right)
\end{equation*}%
is the two-point, complex-time-ordered correlation function of field
operators associated with the quasi-free state $\rho _{\omega }$. See
Section \ref{Section main}. Observe that the matrix in the Pfaffian of (\ref%
{***}) is skew-symmetric, by construction.

\begin{theorem}
\label{thm majorama}\mbox{ }\newline
Let $\{H_{\omega }\}_{\omega \in \Omega }\subset \mathcal{B}\left( \mathfrak{%
h}\right) $ be a $\mathfrak{F}$-measurable\ family $\{H_{\omega }\}_{\omega
\in \Omega }\subset \mathcal{B}\left( \mathfrak{h}\right) $ of bounded
(one-particle) Hamiltonians satisfying Condition \ref{conditon local}. Then,
for all $\omega \in \Omega $, $\beta \in \mathbb{R}^{+}$, $N\in \mathbb{N}$,
$\mathcal{X}=\{x_{1},\ldots ,x_{2N}\}\subset \mathbb{Z}^{d}$ such that $%
\left\vert \mathcal{X}\right\vert =2N$, and $z_{1},\ldots ,z_{2N}\in \mathbb{%
S}_{\beta }$ (see (\ref{strip})),%
\begin{equation*}
\mathbb{E}\left[ \max_{\substack{ p_{1},\ldots ,p_{2N}\in \left\{
0,1\right\}  \\ \sigma _{1},\ldots ,\sigma _{2N}\in \mathrm{S}}}\mathrm{Pf}%
\left[ \mathcal{G}_{\omega }\left( \left( i^{p_{k}}\chi _{I}(H_{\omega })%
\mathfrak{e}_{x_{k},\sigma _{k}},z_{k}\right) ,\left( i^{p_{l}}\chi
_{I}(H_{\omega })\mathfrak{e}_{x_{l},\sigma _{l}},z_{l}\right) \right) %
\right] _{k,l=1}^{2N}\right] \leq 2D\,\mathrm{e}^{-\mu \ell _{\epsilon }(%
\mathcal{X})}.
\end{equation*}%
The constants $\epsilon $, $D$ and $\mu $ are exactly the same as in
Condition \ref{conditon local}.
\end{theorem}

\begin{proof}
The proof uses similar arguments as for determinantal correlation
functionals. We present them in four steps: \medskip

\noindent \underline{Step 1:} Similar to determinants, Pfaffians have a
Laplace expansion with respect to any row of its matrix:
\begin{eqnarray}
\mathrm{Pf}\left[ \mathcal{G}_{\omega }\left( \left( \varphi
_{k},z_{k}\right) ,\left( \varphi _{l},z_{l}\right) \right) \right]
_{k,l=1}^{2N} &=&\sum_{n=1,n\neq m}^{2N}\left( -1\right) ^{m+n+1+\theta
\left( m-n\right) }\mathcal{G}_{\omega }\left( \left( \varphi
_{m},z_{m}\right) ,\left( \varphi _{n},z_{n}\right) \right)   \label{laplace}
\\
&&\times \mathrm{Pf}\left[ \mathcal{G}_{\omega }\left( \left( \varphi
_{k},z_{k}\right) ,\left( \varphi _{l},z_{l}\right) \right) \right]
_{\substack{ k\in \left\{ 1,\ldots ,2N\right\} \backslash \left\{ m\right\}
\\ l\in \left\{ 1,\ldots ,2N\right\} \backslash \left\{ n\right\} }}  \notag
\end{eqnarray}%
for any $\beta \in \mathbb{R}^{+}$, $N\in \mathbb{N}$, $m\in \left\{
1,\ldots ,2N\right\} $, $\varphi _{1},\ldots ,\varphi _{2N}\in \mathfrak{h}$
and $z_{1},\ldots ,z_{2N}\in \mathbb{S}_{\beta }$, where $\theta $ is the
Heaviside step function. Compare (\ref{laplace}) with (\ref{laplace0}%
).\medskip

\noindent \underline{Step 2:} Since $\rho _{\omega }$ is, by definition, a
quasi-free state, observe that%
\begin{equation}
\rho _{\omega }\left( \mathrm{B}\left( \varphi _{1}\right) \cdots \mathrm{B}%
\left( \varphi _{2N}\right) \right) =\mathrm{Pf}\left[ \rho _{\omega }\left(
\mathbb{O}_{\mathrm{id}_{k,l}}\left( \mathrm{B}(\varphi _{k}),\mathrm{B}%
(\varphi _{l})\right) \right) \right] _{k,l=1}^{2N},  \label{totot}
\end{equation}%
for all $N\in \mathbb{N}$ and $\varphi _{1},\ldots ,\varphi _{2N}\in
\mathfrak{h}$, where $\mathrm{id}_{k,l}$ is defined by (\ref{Opi2}), $\pi $
being the neutral element $\mathrm{id}$ of the permutation group $\mathcal{S}%
_{2N}$. See, e.g., \cite[Equations (6.6.9) and (6.6.10)]{EK98}. For any
permutation $\pi \in \mathcal{S}_{2N}$ ($N\in \mathbb{N}$), Equation (\ref%
{totot}) can be written as%
\begin{equation}
\rho _{\omega }%
\Big(%
\mathbb{O}_{\pi }\left( \mathrm{B}\left( \varphi _{1}\right) ,\ldots ,%
\mathrm{B}\left( \varphi _{2N}\right) \right)
\Big)%
=\mathrm{Pf}\left[ \rho _{\omega }\left( \mathbb{O}_{\pi _{k,l}}\left(
\mathrm{B}(\varphi _{k}),\mathrm{B}(\varphi _{l})\right) \right) \right]
_{k,l=1}^{2N},  \label{ass O0-00bis}
\end{equation}%
where $\mathbb{O}_{\pi }$ and the permutation $\pi _{k,l}$ are defined by (%
\ref{Opi1}) and (\ref{Opi2}), respectively. See, e.g., \cite[Proposition B.2]%
{FKT02}. Compare (\ref{ass O0-00bis}) with Lemma \ref{lemma exp copy(2)}%
.\medskip

\noindent \underline{Step 3:} Then, given $2N\in \mathbb{N}$ complex numbers
$z_{1},\ldots ,z_{2N}\in \mathbb{S}_{\beta }$ ($\beta \in \mathbb{R}^{+}$),
similar to (\ref{permutation}), we choose a permutation $\pi \in \mathcal{S}%
_{2N}$ such that, for any $k,l\in \left\{ 1,\ldots ,2N\right\} $, $k\neq l$,
\begin{equation*}
\pi \left( k\right) <\pi \left( l\right) \Leftrightarrow \left[ \mathrm{Im}%
(z_{k})<\mathrm{Im}(z_{l})\right] \vee \left[ \left( \mathrm{Im}(z_{k})=%
\mathrm{Im}(z_{l})\right) \wedge \left( k<l\right) \right] .
\end{equation*}%
Using the Hadamard three-line theorem (via Corollary \ref{determinant bound
copy(2)}), we thus obtain a universal bound on Pfaffians of the form%
\begin{equation}
\left\vert \mathrm{Pf}\left[ \rho _{\omega }\left( \mathbb{O}%
_{z_{k},z_{l}}\left( \tau _{z_{k}}\left( \mathrm{B}(\varphi _{k})\right)
,\tau _{z_{l}}\left( \mathrm{B}(\varphi _{l})\right) \right) \right)
_{k,l=1}^{2N}\right] _{k,l=1}^{N}\right\vert \leq
\prod\limits_{k=1}^{2N}\left\Vert \varphi _{k}\right\Vert _{\mathfrak{h}}
\label{laplace2}
\end{equation}%
for any $N\in \mathbb{N}$, $\varphi _{1},\ldots ,\varphi _{2N}\in \mathfrak{h%
}$ and $z_{1},\ldots ,z_{2N}\in \mathbb{S}_{\beta }$. To get this
inequality, we have used that
\begin{equation*}
\left\Vert \mathrm{B}(\varphi )\right\Vert _{\mathrm{CAR}(\mathfrak{h}%
)}=\left\Vert \varphi \right\Vert _{\mathfrak{h}},\qquad \varphi \in
\mathfrak{h}.
\end{equation*}%
Compare (\ref{laplace2}) with Lemma \ref{determinant bound}. \medskip

\noindent \underline{Step 4:} We infer from (\ref{laplace}) and (\ref%
{laplace2}) that
\begin{equation}
\left\vert \mathrm{Pf}\left[ \mathcal{G}_{\omega }\left( \left( \varphi
_{k},z_{k}\right) ,\left( \varphi _{l},z_{l}\right) \right) \right]
_{k,l=1}^{2N}\right\vert \leq \sum_{n=1,n\neq m}^{2N}\left\vert \mathcal{G}%
_{\omega }\left( \left( \varphi _{m},z_{m}\right) ,\left( \varphi
_{n},z_{n}\right) \right) \right\vert   \label{laplace3}
\end{equation}%
for any $\beta \in \mathbb{R}^{+}$, $N\in \mathbb{N}$, $m\in \left\{
1,\ldots ,2N\right\} $, $\varphi _{1},\ldots ,\varphi _{2N}\in \mathfrak{h}$
and $z_{1},\ldots ,z_{2N}\in \mathbb{S}_{\beta }$. By gauge invariance,
Condition \ref{conditon local} yield the inequality
\begin{equation}
\sum\limits_{x_{2}\in \mathbb{Z}^{d}:\left\vert x_{1}-x_{2}\right\vert
^{\epsilon }\geq R}\mathbb{E}\left[ \sup_{z_{1}z_{2}\in \mathbb{S}_{\beta
}}\max_{\substack{ p_{1},p_{2}\in \left\{ 0,1\right\}  \\ \sigma _{1},\sigma
_{2}\in \mathrm{S}}}\left\vert \mathcal{G}_{\omega }\left( \left( i^{p_{1}}%
\mathfrak{e}_{x_{1},\sigma _{1}},z_{1}\right) ,\left( i^{p_{2}}\mathfrak{e}%
_{x_{2},\sigma _{2}},z_{2}\right) \right) \right\vert \right] \leq 2D\,%
\mathrm{e}^{-\mu R}.  \label{laplace4}
\end{equation}%
Therefore, the assertion is a direct consequence of Inequalities (\ref%
{laplace3}) and (\ref{laplace4}).
\end{proof}

Theorem \ref{thm majorama} is a version of \cite[Theorem 1.4]{Sims-Warzel}
which holds true at any dimension $d\in \mathbb{N}$ and for any complex
times within the strip $\mathbb{S}_{\beta }$. A result similar to \cite[%
Theorem 1.3]{Sims-Warzel} for the many-point correlation functions of field
operators at fixed $\omega \in \Omega $, instead of an estimate for their
expectation values, easily follows by replacing Condition \ref{conditon
local} with a similar bound for a fixed $\omega \in \Omega $. We again omit
the details.

One observation in relation with \cite[Theorems 1.3 and 1.4]{Sims-Warzel} is
important to mention: For any disjoint partition $\mathcal{X}_{1},\mathcal{X}%
_{2}$ of $\mathcal{X}\subset \mathbb{Z}^{d}$, we deduce from (\ref{Hausdorf}%
) and (\ref{distance concentration}) that%
\begin{equation}
\ell _{\epsilon }(\mathcal{X})\leq \mathfrak{d}_{\epsilon }(\mathcal{X}_{1},%
\mathcal{X}_{2})\ .
\end{equation}%
By (\ref{hadamabis}) it follows that, for any disjoint partition $\mathcal{X}%
_{1},\mathcal{X}_{2}$ of $\mathcal{X}\subset \mathbb{Z}^{d}$ such that $|%
\mathcal{X}_{1}|=|\mathcal{X}_{2}|$,
\begin{equation*}
\ell _{\epsilon }(\mathcal{X})\leq \mathfrak{d}_{\epsilon }^{(S)}(\mathcal{X}%
_{1},\mathcal{X}_{2})\ .
\end{equation*}%
Note that the right hand side of the above inequality corresponds to \cite[%
Equation (1.27)]{Sims-Warzel} when $\mathcal{X}_{1}=\{x_{1},x_{3},\ldots
,x_{2N-1}\}\subset \mathbb{Z}^{d=1}$ and $\mathcal{X}_{2}=\{x_{2},x_{4},%
\ldots ,x_{2N}\}\subset \mathbb{Z}^{d=1}$ for $2N\in \mathbb{N}$ (different)
lattice points. Therefore, our notion of localization for Pfaffian
correlation functionals is weaker than the one used in \cite[Theorems 1.3
and 1.4]{Sims-Warzel}. Note, however, that, like $\ell _{\epsilon }(\mathcal{%
X})$, the quantity \cite[Equation (1.27)]{Sims-Warzel} stays small if the
points of $\mathcal{X}$ are packed in clusters containing exactly two points
$\{x_{k},x_{k+1}\}$, $k=1,3,\ldots ,2N$, independently of how far-apart from
each other the clusters are. Therefore, our notion of localization captures
qualitatively the behavior of the one used in \cite[Theorems 1.3 and 1.4]%
{Sims-Warzel}.

\section{Appendix: Log convexity of Multivariable Analytic Functions on
Tubes \label{Section Hadamard}}

Fix $\beta \in \mathbb{R}^{+}$. Let
\begin{equation*}
\mathfrak{T}_{1}\doteq \left\{ \xi \in \mathbb{C}:\mathrm{Im}\left\{ \xi
\right\} \in \left[ -\beta ,0\right] \right\} =\mathbb{S}_{\beta },
\end{equation*}%
(see (\ref{strip})) and $f:\mathfrak{T}_{1}\rightarrow \mathbb{C}$ be a
bounded continuous function. Define the map $B_{f}^{(1)}:\left[ -\beta ,0%
\right] \rightarrow \lbrack -\infty ,\infty )$ by%
\begin{equation*}
B_{f}^{(1)}(s)\doteq \ln \left( \sup_{t\in \mathbb{R}}\left\vert f\left(
t+is\right) \right\vert \right) .
\end{equation*}%
We use the convention $\ln 0\doteq -\infty $ and $0\cdot (-\infty )\doteq
-\infty $. Then, the Hadamard three-line theorem \cite[Theorem 12.3]{Simon}
states:

\begin{theorem}
\label{determinant bound copy(1)}\mbox{ }\newline
Let $\beta \in \mathbb{R}^{+}$ and $f:\mathfrak{T}_{1}\rightarrow \mathbb{C}$
be a bounded continuous function. If $f$ is holomorphic in the interior of $%
\mathfrak{T}_{1}$ then $B_{f}^{(1)}$ is a convex function.
\end{theorem}

This theorem has the following generalization to holomorphic functions in
several variables: For all $n\in \mathbb{N}$, let $K_{n}\subset \mathbb{R}%
^{n}$ be the simplex
\begin{equation*}
K_{n}\doteq \left\{ \left( s_{1},\ldots ,s_{n}\right) :s_{1},\ldots
,s_{n}\in \left[ -\beta ,0\right] ,\text{ }s_{1}+\cdots +s_{n}\geq -\beta
\right\}
\end{equation*}%
and define, for all $n\in \mathbb{N}$, the \textquotedblleft
tube\textquotedblright\
\begin{equation}
\mathfrak{T}_{n}\doteq \left\{ \left( \xi _{1},\ldots ,\xi _{n}\right) \in
\mathbb{C}^{n}:\left( \mathrm{Im}\left\{ \xi _{1}\right\} ,\ldots ,\mathrm{Im%
}\left\{ \xi _{n}\right\} \right) \in K_{n}\right\} .  \label{tubetueb}
\end{equation}%
Define further the map $B_{f}^{(n)}:K_{n}\rightarrow \lbrack -\infty ,\infty
)$ by%
\begin{equation*}
B_{f}^{(n)}(s_{1},\ldots ,s_{n})\doteq \ln \left( \sup_{\left( t_{1},\ldots
,t_{n}\right) \in \mathbb{R}^{n}}\left\vert f\left( t_{1}+is_{1},\ldots
,t_{n}+is_{n}\right) \right\vert \right)
\end{equation*}%
with $f:\mathfrak{T}_{n}\rightarrow \mathbb{C}$ being a bounded continuous
function. Then, we obtain the following corollary:

\begin{corollary}
\label{determinant bound copy(2)}\mbox{ }\newline
Let $\beta \in \mathbb{R}^{+}$, $n\in \mathbb{N}$ and $f:\mathfrak{T}%
_{n}\rightarrow \mathbb{C}$ be a bounded continuous function. If $f$ is
holomorphic in the interior of $\mathfrak{T}_{n}$ then $B_{f}^{(n)}$ is a
convex function.
\end{corollary}

\begin{proof}
Fix all parameters of the corollary and assume that $f$ is holomorphic in
the interior of $\mathfrak{T}_{n}$. Take $(s_{1},\ldots ,s_{n})\in K_{n}$
and $(s_{1}^{\prime },\ldots ,s_{n}^{\prime })\in K_{n}$. For all $%
(t_{1},\ldots ,t_{n})\in \mathbb{R}^{n}$, define the function $%
F_{(t_{1},\ldots ,t_{n})}:\mathfrak{T}_{1}\rightarrow \mathbb{C}$ by%
\begin{equation*}
F_{(t_{1},\ldots ,t_{n})}\left( \xi \right) \doteq f\left(
t_{1}+i(s_{1}(1+\xi \beta ^{-1})-s_{1}^{\prime }\xi \beta ^{-1}),\ldots
,t_{n}+i(s_{n}(1+\xi \beta ^{-1})-s_{n}^{\prime }\xi \beta ^{-1})\right)
\text{ }.
\end{equation*}%
For all $\xi \in \mathfrak{T}_{1}$, note that%
\begin{equation*}
\left( t_{1}+i(s_{1}(1+\xi \beta ^{-1})-s_{1}^{\prime }\xi \beta
^{-1}),\ldots ,t_{n}+i(s_{n}(1+\xi \beta ^{-1})-s_{n}^{\prime }\xi \beta
^{-1})\right) \in \mathfrak{T}_{n}\text{ },
\end{equation*}%
by convexity of $K_{n}$. This function is bounded and continuous on $%
\mathfrak{T}_{1}$, and holomorphic in the interior of $\mathfrak{T}_{1}$.
Hence, by Theorem \ref{determinant bound copy(1)}, for all $\alpha \in
\lbrack 0,1]$,%
\begin{eqnarray}
\ln \left( \sup_{t\in \mathbb{R}}\left\vert F_{(t_{1},\ldots ,t_{n})}\left(
t-i\alpha \beta \right) \right\vert \right) &\leq &\alpha \ln \left(
\sup_{t\in \mathbb{R}}\left\vert F_{(t_{1},\ldots ,t_{n})}\left( t-i\beta
\right) \right\vert \right)  \label{mono} \\
&&+(1-\alpha )\ln \left( \sup_{t\in \mathbb{R}}\left\vert F_{(t_{1},\ldots
,t_{n})}\left( t\right) \right\vert \right) \text{ }.  \notag
\end{eqnarray}%
Since $\ln $ is a monotonically increasing, continuous function, for all $%
\alpha \in \lbrack 0,1]$,
\begin{eqnarray*}
B_{f}^{(n)}\left( \alpha s_{1}^{\prime }+(1-\alpha )s_{1},\ldots ,\alpha
s_{n}^{\prime }+(1-\alpha )s_{n}\right) &=&\ln \left( \sup_{(t_{1},\ldots
,t_{n})\in \mathbb{R}^{n}}\sup_{t\in \mathbb{R}}\left\vert F_{(t_{1},\ldots
,t_{n})}\left( t-i\alpha \beta \right) \right\vert \right) \\
&=&\sup_{(t_{1},\ldots ,t_{n})\in \mathbb{R}^{n}}\ln \left( \sup_{t\in
\mathbb{R}}\left\vert F_{(t_{1},\ldots ,t_{n})}\left( t-i\alpha \beta
\right) \right\vert \right) \text{ },
\end{eqnarray*}%
which, by (\ref{mono}), in turn implies that
\begin{equation*}
B_{f}^{(n)}\left( \alpha s_{1}^{\prime }+(1-\alpha )s_{1},\ldots ,\alpha
s_{n}^{\prime }+(1-\alpha )s_{n}\right) \leq \left( 1-\alpha \right)
B_{f}^{(n)}\left( s_{1},\ldots ,s_{n}\right) +\alpha B_{f}^{(n)}\left(
s_{1}^{\prime },\ldots ,s_{n}^{\prime }\right)
\end{equation*}%
for all $\alpha \in \lbrack 0,1]$.
\end{proof}

\bigskip \noindent \textbf{Acknowledgements: }This research is supported by
the\ FAPESP, the CNPq, the Basque Government through the grant IT641-13 and
the BERC 2014-2017 program, and by the Spanish Ministry of Economy and
Competitiveness MINECO: BCAM Severo Ochoa accreditation SEV-2013-0323 and
MTM2014-53850. We are very grateful to the BCAM and its management, which
supported this project via the visiting researcher program. Finally, we
thank S. Warzel for discussions.


\begin{thebibliography}{99}
\bibitem{interaction0} D.M. Basko, I.L. Aleiner, B.L. Altshuler,
Metal-insulator transition in a weakly interacting many-electron system with
localized single-particle states. \textit{Ann. Phys.} \textbf{321} (2006)
1126--1205

\bibitem{interaction2} V. Chulaevsky and Y.M. Suhov, Eigenfunctions in a
Two-Particle Anderson Tight Binding Model, \textit{Commun. Math. Phys.}
\textbf{289} (2009) 701--723

\bibitem{Aiz-Warzel} M. Aizenman and S. Warzel, Localization Bounds for
Multiparticle Systems, \textit{Commun. Math. Phys.} \textbf{290}(3) (2009)
903--934

\bibitem{Sims-Warzel} R. Sims and S. Warzel, Decay of Determinantal and
Pfaffian Correlation Functionals in One-Dimensional Lattices, \textit{%
Commun. Math. Phys.} \textbf{347} (2016) 903--931; Erratum, \textit{Commun.
Math. Phys.} In preparation (2017).

\bibitem{cite1} H. Abdul-Rahmana and G. Stolz, A uniform area law for the
entanglement of eigenstates in the disordered XY chain, \textit{J. Math.
Phys.} \textbf{56} (2015) 121901

\bibitem{cite2} R. Seiringer and S. Warzel, Decay of correlations and
absence of superfluidity in the disordered Tonks--Girardeau gas, \textit{New
J. Phys.} \textbf{18} (2016) 035002 (14pp)

\bibitem{cite3} M. Gebert and M. Lemm, On Polynomial Lieb--Robinson Bounds
for the XY Chain in a Decaying Random Field \textit{J Stat Phys} \textbf{164}%
\ (2016) 667--679

\bibitem{cite4} H. Abdul-Rahman, B. Nachtergaele, R. Sims, G. Stolz,
Entanglement Dynamics of Disordered Quantum XY Chains. \textit{Lett. Math.
Phys.} \textbf{106}(5) (2016) 649--674.

\bibitem{cite5} G. Stolz, Many-body localization for disordered Bosons,
\textit{New J. Phys.} \textbf{18} (2016) 031002 (3 pp)

\bibitem{cite6} R. Mavi, Localization for the Ising model in a transverse
field with generic aperiodic disorder, \qquad arXiv:1605.06514 (2016)

\bibitem{cite7} H. Abdul-Rahman, B. Nachtergaele, R. Sims, G. Stolz,
Localization properties of the disordered XY spin chain, Ann. Phys. (Berlin)
1600280 (2017) (17pp)

\bibitem{cite8} V. Beaud and S. Warzel, Low-energy Fock-space localization
for attractive hard-core particles in disorder, arXiv:1703.02465 (2017)

\bibitem{cite9} P. Le Doussal, S.N. Majumdar, G. Schehr, Periodic Airy
process and equilibrium dynamics of edge fermions in a trap,
arXiv:1702.06931 (2017)

\bibitem{BratteliRobinson} O. Bratteli and D.W. Robinson, \textit{Operator
Algebras and Quantum Statistical Mechanics, Vol. II, 2nd ed.}
Springer-Verlag, New York, 1996

\bibitem{WAlivre} M. Aizenman and S. Warzel, \textit{Random Operators:
Disorder Effects on Quantum Spectra and Dynamics}, GSM 168, 2016

\bibitem{EK98} D.E. Evans and Y. Kawahigashi, \textit{Quantum Symmetries on
Operator Algebras}, Oxford Mathematical Monographs, Oxford University Press,
1998

\bibitem{FKT02} J. Feldman, H. Knorrer and E. Trubowitz, Fermionic
Functional Integrals and the Renormalization Group, CRM Monograph Series,
Volume 16 (American Mathematical Society, 2002).

\bibitem{Bru-Pedra-deter-bound} J.-B. Bru and W. de Siqueira Pedra,
Universal Bounds for Large Determinants from Non--Commutative H\"{o}lder
Inequalities in Fermionic Constructive Quantum Field Theory, M3AS:
Mathematical Models and Methods in Applied Sciences \textbf{12}(27) (2017)

\bibitem{Araki-Masuda} H. Araki and T. Masuda: Positive Cones and $L_{p}$%
-Spaces for von Neumann Algebras, \textit{Publ. RIMS, Kyoto Univ.}\textbf{18}
(1982) 339--411

\bibitem{Simon} B. Simon, \textit{Convexity: An Analytic Viewpoint,}
Cambridge Tracts in Mathematics 187, Cambridge University Press, 2011
\end{thebibliography}
\end{document}